\documentclass[11pt,letter]{amsart}
\usepackage[text={4.5in,7in},centering]{geometry}
\makeatletter

\usepackage{graphicx}
\usepackage{amsmath}
\usepackage{color}
\usepackage{subfigure}

\usepackage{algorithmic}
\usepackage{algorithm}

\usepackage{pgf,tikz}
\usetikzlibrary{arrows}
\usetikzlibrary{snakes}

\usepackage{subfig} 

\title{Packing and Covering a Polygon with Geodesic Disks}

\author{Ivo Vigan}

\thanks{Dept. of Computer Science, City University of New York, The Graduate Center, New York, NY, USA. {\tt ivigan@gc.cuny.edu}}
\thanks{Research supported by NSF grant 1017539} 

\theoremstyle{plain}
\newtheorem{theorem}{Theorem}
\newtheorem{lemma}[theorem]{Lemma}
\newtheorem{prop}[theorem]{Proposition}

\newtheorem{problem}[theorem]{Problem}
\newtheorem{definition}[theorem]{Definition}

\newtheorem{obs}[theorem]{Observation}

\begin{document}
\maketitle

\begin{abstract}
Given a polygon $P$, for two points $s$ and $t$ contained in the polygon, their \emph{geodesic distance} is the length of the shortest $st$-path within $P$. A \emph{geodesic disk} of radius $r$ centered at a point $v \in P$ is the set of points in $P$ whose geodesic distance to $v$ is at most $r$. We present a polynomial time $2$-approximation algorithm for finding a densest geodesic unit disk packing in $P$. Allowing arbitrary radii but constraining the number of disks to be $k$, we present a $4$-approximation algorithm for finding a packing in $P$ with $k$ geodesic disks whose minimum radius is maximized. We then turn our focus on $\emph{coverings}$ of $P$ and present a $2$-approximation algorithm for covering $P$ with $k$ geodesic disks whose maximal radius is minimized. Furthermore, we show that all these problems are $\mathsf{NP}$-hard in polygons with holes. Lastly, we present a polynomial time exact algorithm which covers a polygon with two geodesic disks of minimum maximal radius.
\end{abstract}

\section{Motivation and Related Work}

Packing and covering problems are among the most studied problems in discrete geometry (see \cite{agar},\cite{aste},\cite{boro},\\ \cite{erd},\cite{convex},\cite{grub},\cite{rog},\cite{zusch},\cite{toth1993packing},\cite{recentT},\cite{toth},\cite{zong} for books on these topics). Nevertheless, most of the literature focus on packings and coverings using Euclidean balls, which is a somewhat unrealistic assumption for practical problems. A prominent practical example is the \emph{Facility Location} (\emph{$k$-Center}) problem (see for example \cite{drezner1995facility}) in buildings or other constrained areas. In this a setting the relevant distance metric is the shortest path metric and not the Euclidean distance. Such a problem occurs when a mobile robot is navigating in a room such as a data center (see \cite{LenchnerICAC} and \cite{5980554} for such an example), which is naturally modeled as a polygon, and we are interested in placing charging stations in such a way that the worst case travel time of the robot to the closest station gets minimized \cite{ibmRob}.\\
The shortest path distance is also referred to as the \emph{geodesic distance} and, for two points $u$ and $v$ in a polygon $P$, it is denoted by $d(u,v)$ and defined as the length of the shortest path between $u$ and $v$ which stays inside $P$. Furthermore, we define a closed \emph{geodesic disk} $D$ of radius $r$ centered at a point $v \in P$, as the set of all points in $P$ whose geodesic distance to $v$ is at most $r$. The \emph{interior} of $D$, denoted by $int(D)$, contains all points in $P$ which are at distance less than $r$ from $v$. The \emph{boundary} of $D$, denoted by $\partial D$, contains all points of $P$ which are either exactly at a distance $r$ from $v$ or they are at distance at most $r$ from $v$ but contained in the polygon boundary $\partial P$ (see Figure \ref{geoDisk}). 

In this paper we would like to initiate the studies of packing and covering problems in polygons using geodesic disks.\\

For packing problems in polygons, several complexity theoretical results are known. In  \cite{packNard} it is shown that  packing unit squares into orthogonal polygons with holes is $\mathsf{NP}$-hard, while the complexity is still open for simple polygons \cite{OPP}. On the other hand, in \cite{iacono} it is shown that finding the maximum number of small polygons which can be packed into a simple polygon is $\mathsf{NP}$-hard. This result was improved in \cite{ningPack} where hardness is shown, even if the small polygons have constant size. In case where the number of disks to be packed is fixed, the problem is known as a \emph{Dispersion} or \emph{Obnoxious Facility Location} problem and its Euclidean versions have been studied in \cite{dispersion1980} and \cite{Erkut1989275} while other distance functions were considered in \cite{hosseini2009obnoxious}.

\begin{figure}
\center
\definecolor{uuuuuu}{rgb}{0.27,0.27,0.27}
\definecolor{qqqqff}{rgb}{0.9,0.9,.9}

\definecolor{dddddd}{rgb}{0.7,0.6,0.4}

\begin{tikzpicture}[line cap=round,line join=round,>=triangle 45,x=.7cm,y=.7cm]
\clip(1.5,-0.5) rectangle (6.4,6.1);
\fill[color=qqqqff,fill=qqqqff,fill opacity=1.0] (2.73,3.12) -- (1.9,2.66) -- (2.01,3.18) -- cycle;
\fill[color=qqqqff,fill=qqqqff,fill opacity=1.0] (4,3) -- (2.01,3.18) -- (2.27,4.4) -- cycle;
\fill[color=qqqqff,fill=qqqqff,fill opacity=1.0] (4,3) -- (4.03,2.47) -- (2.73,3.12) -- cycle;
\fill[color=qqqqff,fill=qqqqff,fill opacity=1.0] (4,3) -- (3.71,3.73) -- (3.03,5) -- cycle;
\fill[color=qqqqff,fill=qqqqff,fill opacity=1.0] (4,3) -- (4.46,0.83) -- (4.03,2.47) -- cycle;
\fill[color=qqqqff,fill=qqqqff,fill opacity=1.0] (4,3) -- (4.83,0.94) -- (4.98,3.82) -- cycle;
\fill[color=qqqqff,fill=qqqqff,fill opacity=1.0] (4,3) -- (5.16,4.89) -- (3.71,3.73) -- cycle;

\draw [shift={(4,3)},color=qqqqff,fill=qqqqff,fill opacity=1.0]  (0,0) --  plot[domain=4.92:5.1,variable=\t]({1*2.22*cos(\t r)+0*2.22*sin(\t r)},{0*2.22*cos(\t r)+1*2.22*sin(\t r)}) -- cycle ;
\draw [shift={(4,3)},color=qqqqff,fill=qqqqff,fill opacity=1.0]  (0,0) --  plot[domain=0.7:1.02,variable=\t]({1*2.22*cos(\t r)+0*2.22*sin(\t r)},{0*2.22*cos(\t r)+1*2.22*sin(\t r)}) -- cycle ;
\draw [shift={(4.98,3.82)},color=qqqqff,fill=qqqqff,fill opacity=1.0]  (0,0) --  plot[domain=-0.42:0.7,variable=\t]({1*0.94*cos(\t r)+0*0.94*sin(\t r)},{0*0.94*cos(\t r)+1*0.94*sin(\t r)}) -- cycle ;

\draw [shift={(4,3)},color=qqqqff,fill=qqqqff,fill opacity=1.0]  (0,0) --  plot[domain=2.02:2.46,variable=\t]({1*2.22*cos(\t r)+0*2.22*sin(\t r)},{0*2.22*cos(\t r)+1*2.22*sin(\t r)}) -- cycle ;

\draw [shift={(4.98,3.82)}] plot[domain=-0.42:0.7,variable=\t]({1*0.94*cos(\t r)+0*0.94*sin(\t r)},{0*0.94*cos(\t r)+1*0.94*sin(\t r)});
\draw [shift={(4,3)}] plot[domain=4.92:5.1,variable=\t]({1*2.22*cos(\t r)+0*2.22*sin(\t r)},{0*2.22*cos(\t r)+1*2.22*sin(\t r)});
\draw [shift={(4,3)}] plot[domain=2.02:2.46,variable=\t]({1*2.22*cos(\t r)+0*2.22*sin(\t r)},{0*2.22*cos(\t r)+1*2.22*sin(\t r)});
\draw [shift={(4,3)}] plot[domain=0.7:1.02,variable=\t]({1*2.22*cos(\t r)+0*2.22*sin(\t r)},{0*2.22*cos(\t r)+1*2.22*sin(\t r)});

\draw [shift={(2.73,3.12)},color=qqqqff,fill=qqqqff,fill opacity=1.0]  (0,0) --  plot[domain=3.65:4.79,variable=\t]({1*0.95*cos(\t r)+0*0.95*sin(\t r)},{0*0.95*cos(\t r)+1*0.95*sin(\t r)}) -- cycle ;

\draw [shift={(2.73,3.12)}] plot[domain=3.65:4.79,variable=\t]({1*0.95*cos(\t r)+0*0.95*sin(\t r)},{0*0.95*cos(\t r)+1*0.95*sin(\t r)});

\draw[dddddd] (2.88,0.17)-- (1.67,1.56);
\draw[dddddd] (1.67,1.56)-- (2.58,5.83);
\draw[dddddd] (2.58,5.83)-- (3.71,3.73);
\draw[dddddd] (3.71,3.73)-- (6.06,5.61);
\draw[dddddd] (6.06,5.61)-- (6.19,3.29);
\draw[dddddd] (6.19,3.29)-- (4.98,3.82);
\draw[dddddd] (4.98,3.82)-- (4.76,-0.34);
\draw[dddddd] (4.76,-0.34)-- (4.03,2.47);
\draw[dddddd] (4.03,2.47)-- (2.73,3.12);
\draw[dddddd] (2.73,3.12)-- (2.83,1.9);
\draw[dddddd] (2.83,1.9)-- (2.88,0.17);

\draw [color=qqqqff] (2.73,3.12)-- (1.9,2.66);
\draw [color=black] (1.9,2.66)-- (2.01,3.18);
\draw [color=qqqqff] (2.01,3.18)-- (2.73,3.12);
\draw [color=qqqqff] (4,3)-- (2.01,3.18);
\draw [color=black] (2.01,3.18)-- (2.27,4.4);
\draw [color=qqqqff] (2.27,4.4)-- (4,3);
\draw [color=qqqqff] (4,3)-- (4.03,2.47);
\draw [color=black] (4.03,2.47)-- (2.73,3.12);
\draw [color=qqqqff] (2.73,3.12)-- (4,3);
\draw [color=qqqqff] (4,3)-- (3.71,3.73);

\draw [color=qqqqff] (3.03,5)-- (4,3);
\draw [color=qqqqff] (4,3)-- (4.46,0.83);
\draw [color=black] (4.46,0.83)-- (4.03,2.47);
\draw [color=qqqqff] (4.03,2.47)-- (4,3);
\draw [color=qqqqff] (4,3)-- (4.83,0.94);
\draw [color=black] (4.83,0.94)-- (4.98,3.82);
\draw [color=qqqqff] (4.98,3.82)-- (4,3);
\draw [color=qqqqff] (4,3)-- (5.16,4.89);
\draw [color=black] (5.16,4.89)-- (3.71,3.73);
\draw [color=qqqqff] (3.71,3.73)-- (4,3);
\draw [color=black] (4.98,3.82) -- (5.84,3.44);
\draw [color=black](2.73,3.12)  -- (2.81,2.17) ;
\draw [color=black] (3.71,3.73)-- (3.03,5);

\begin{scriptsize}
\fill [color=black] (2.88,0.17) circle (1.5pt);
%\draw[color=black] (2.98,0.33) node {$A$};
\fill [color=black] (1.67,1.56) circle (1.5pt);
%\draw[color=black] (1.75,1.72) node {$B$};
\fill [color=black] (2.58,5.83) circle (1.5pt);
%\draw[color=black] (2.68,5.99) node {$C$};
\fill [color=black] (3.71,3.73) circle (1.5pt);
%\draw[color=black] (3.8,3.89) node {$D$};
\fill [color=black] (6.06,5.61) circle (1.5pt);
%\draw[color=black] (6.14,5.76) node {$E$};
\fill [color=black] (6.19,3.29) circle (1.5pt);
%\draw[color=black] (6.28,3.45) node {$F$};
\fill [color=black] (4.98,3.82) circle (1.5pt);
%\draw[color=black] (5.08,3.97) node {$G$};
\fill [color=black] (4.76,-0.34) circle (1.5pt);
%\draw[color=black] (4.86,-0.19) node {$H$};
\fill [color=black] (4.03,2.47) circle (1.5pt);
%\draw[color=black] (4.09,2.63) node {$I$};
\fill [color=black] (2.73,3.12) circle (1.5pt);
%\draw[color=black] (2.81,3.28) node {$J$};
%\fill [color=black] (2.83,1.9) circle (1.5pt);
%\draw[color=black] (2.93,2.05) node {$K$};
\fill [color=black] (4,3) circle (1.5pt);
\draw[color=black] (4.08,3.3) node {$v$};
\end{scriptsize}
\end{tikzpicture}
\caption{A polygon containing a geodesic disk centered at $v$, whose interior is depicted in gray and its boundary is drawn in black.}
\label{geoDisk}
\end{figure}
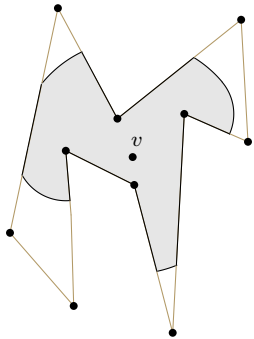

Covering problems have been studied in the context of the \emph{Metric $k$-Center Clustering} problem, where $n$ points in the plane are covered with $k$ metric disks of minimum maximal radius. In \cite{kcenterRef} it was shown to be $\mathsf{NP}$-hard and a $2$-approximation algorithm was presented, which is the best possible approximation ratio when allowing arbitrary metrics \cite{Vazirani}. For the Euclidean metric it is shown in \cite{2approxKc} to be inapproximable in polynomial time within a factor of $1.82$ unless $\mathsf{P} = \mathsf{NP}$. In the context of polygons, \cite{convexCOv} present an $(0.78-\epsilon)$-approximation algorithm for covering a convex polygon with $k$ Euclidean disks, under the restriction that the disks need to be fully contained within the polygon.

% In robotics \cite{2008Robot} and \cite{2013Robot} study ...

Exact coverings of $n$ points in the plane with two Euclidean disks of minimum maximal radius, commonly referred to as the \emph{2-Center} problem, has been heavily studied. The best deterministic algorithm runs in $O(n \log^9 n)$ time \cite{Sharir:1996:NAP:237218.237251} and in \cite{Eppstein1997} an expected $O(n \log^2 n)$ time algorithm is presented. For polygons, in \cite{Shin98computingtwo} a $O(n^2\log^3 n)$ time algorithm for covering a convex polygon with two Euclidean disks of minimum maximal radius is presented.

This paper is organized as follows. In Section \ref{udPack} we present a polynomial time $2$-approximation algorithm for packing geodesic unit disks into a simple polygon and show that the problem is $\mathsf{NP}$-hard in polygons with holes. In Section \ref{udCover} we show that covering a polygon with holes using geodesic unit disks is $\mathsf{NP}$-hard. In Section \ref{geodesicKcov} we allow arbitrary radii, present a $2$-approximation algorithm for covering a polygon (possibly with holes) with $k$ disks of minimum maximal radius and show that the problem is $\mathsf{NP}$-hard in polygons with holes. Analogously, in Section \ref{geoKpacking} we present a $4$-approximation algorithm for finding a packing in a polygon with $k$ geodesic disks whose minimum radius is maximized and show that the problem is $\mathsf{NP}$-hard in polygons with holes. Finally, in Section \ref{exact2cov} we present a polynomial time exact algorithm which covers a simple polygon with two geodesic disks of minimum maximal radius.

\section{Geodesic Unit Disk Packing}
\label{udPack}

In this section we study the following Geodesic Unit Disk Packing problem, present a $2$-approximation algorithm for simple polygons and show that the problem is $\mathsf{NP}$-hard in polygons with holes. 

\begin{problem}[Geodesic Unit Disk Packing]
Given a polygon, find a maximum cardinality packing with geodesic disks of radius $1$. 
\label{packproblem}
\end{problem}

%Since in a collection of closed disks of radius $r$, disks do not contain the center of any other disk if and only if the disks of radius $r/2$ with the same centers are disjoint, the Geodesic Disk Packing problem can be equivalently stated as the following problem.

%\begin{problem}[Geodesic Unit Point Packing]
%Given a polygon $P$, find a maximum cardinality set of points in $P$ under the condition that the minimum geodesic distance between any two points is at least one. 
%\label{packproblem}
%\end{problem}	

%Furthermore, we let $diam(P) = \max_{u,v \in P} d(u,v)$ be the \emph{geodesic diameter} of $P$ and we refer to two points $u,v$ with $d(u,v) = diam(P)$ as a \emph{diametral pair} and $u,v$ are each called a \emph{diametral vertex}. 

\begin{algorithm}[htp]
\caption{greedyUnitPacking($P$)}
\label{alg1}

\begin{algorithmic}
\STATE $S \leftarrow \emptyset$ \COMMENT{The centers of the packing.}
\STATE $\partial \mathcal{A} \leftarrow \emptyset$  \COMMENT{The boundary of the geodesic disk arrangement.}
\STATE $V \leftarrow vertices(P)$  \COMMENT{The set of points at which disks	 can be centered; initially these are the vertices of $P$.}
\REPEAT
\STATE find $u,v \in V$ of maximum geodesic distance
\STATE  center a geodesic disk $D$ of radius $2$ at $v$
\STATE $S \leftarrow S \cup \{v\}$ \COMMENT{Add $v$ to the set of centers returned by the algorithm.}
\STATE $V \leftarrow V \setminus int(D)$ \COMMENT{Remove all points of $V$ contained in the interior of $D$.}
\STATE Update $V$ by including all intersection points of $\partial D$ with $\partial \mathcal{A}$ and with the unpacked portion of $\partial P$. 
\STATE $\partial \mathcal{A} \leftarrow \partial \mathcal{A} \cup \partial D$ \COMMENT{Update the boundary of the disk arrangement by including $\partial D$.}\\
\UNTIL{$V = \emptyset$}
\RETURN $S$
\end{algorithmic}
\end{algorithm}

The algorithm greedily centers geodesic disks of radius $2$ at points which are part of a maximum distance pair in the currently unpacked (i.e. free) region of the polygon. The reasoning behind placing radius $2$ and not unit disks is that a radius $2$ disk does not contain the center of any other radius $2$ disk if and only if the disks of radius $1$ with the same centers are disjoint, i.e. form a packing. The set of candidate points of maximal distance, denoted by $V$ in the algorithm, is initialized to consist of all the polygon vertices. In each iteration, the intersection points of the boundary  $\partial D$ of the newly placed disk with unpacked parts of the polygon boundary as well as the intersection points of the newly placed disk with the boundary $\partial \mathcal{A}$ of the disk arrangement  induced by the previously placed disks are added to $V$. Furthermore, all points of $V$ lying in the interior of the newly placed disk $D$ get removed from $V$.

\begin{prop}
The \emph{greedyUnitPacking} algorithm runs in time $O(K(n+K)\log^2(n+K))$, where $n$ is the number of vertices of the polygon and $K$ is the size of the output. 
\end{prop}
\begin{proof}
Using the algorithm introduced in \cite{Guibas} (see also \cite{geodPac}), the boundary of a geodesic disk can be computed in time $O(n)$. Furthermore, using the algorithm of \cite{diskBoundShar}, one can compute the boundary of an arrangement of $K$ geodesic disks in $O((n + K)\log^2(n+K))$ time. It is easy to see that the cardinality of $V$ is $O(n + K)$ at any step, thus finding the next center $v$ among the set $V$ of points, i.e. finding a point of a maximum distance pair in the free space, can be computed in $O((n + K)\log(n+K))$ time using the algorithm of \cite{Toussaint89computinggeodesic}. Since the main loop runs $K$ times, the proposition follows.
\end{proof}

\begin{theorem}
The greedyUnitPacking algorithm yields a $2$-approximation for the Geodesic Unit Disk Packing problem.  % and runs in time $O(kn)$, where $k$ denotes the size of the solution $|S|$.
\end{theorem}

\begin{proof}
Since at each step $greedyUnitPacking$ centers the newly placed disk at a point in $V$ and all these points are all at least at distance $2$ from any point in $S$, the computed centers indeed induce a geodesic unit disk packing. In order to see that $S$ has at least half the cardinality of an optimal packing $OPT$, let $(p_1, \ldots, p_K)$ be the sequence of the points $S$ as placed by $greedyUnitPacking$. For $1 \leq j \leq K$ let $D_j$ be the geodesic radius $2$ disk centered at $p_j$. Letting $free_i(P) = P \setminus  \bigcup^{i-1}_{j=1} int(D_j)$ denote the free (i.e. unpacked) regions of $P$ at the beginning of the $i$-th iteration, the algorithm selects a point $p_i$ of $V$ which is a diametral point in $free_i(P)$. Thus Lemma \ref{distLemma} implies that the part of $D_i$ which is contained in $free_i(P)$ contains at most two centers of any optimal solution, i.e. $|free_i(P) \cap int(D_i) \cap OPT| \leq 2$. Since $\bigcup^{K}_{j=1} int(D_j) \cap free_j(P)$ covers all of $P$, it holds for each $c \in OPT$ that there is a disk $D_i$, such that $c \in free_i(P) \cap int(D_i)$. Thus $\{int(D_1) \cap free_1(P), \ldots, int(D_K) \cap free_K(P) \}$ \emph{partitions} $OPT$ into $\leq K$ blocks each of size $\leq 2$ implying $|OPT| \leq 2K$. 
\end{proof}

\begin{figure}
\center
\begin{tikzpicture}[line cap=round,line join=round,>=triangle 45,x=.45cm,y=.45cm]
\clip(0,-6) rectangle (8.3,5.9);
\draw [color=black] (0.6,2.04)-- (2.34,2.8);
\draw [color=black] (2.34,2.8)-- (3.26,4.02);
\draw [color=black] (3.26,4.02)-- (3.78,5.12);
\draw [color=black] (3.78,5.12)-- (4.3,3.82);
\draw [color=black] (4.3,3.82)-- (5.74,2.42);
\draw [color=black] (5.74,2.42)-- (7.86,1.66);
\draw [color=black] (7.86,1.66)-- (6.88,1.32);
\draw [color=black] (6.88,1.32)-- (6.2,0.9);
\draw [color=black] (6.2,0.9)-- (6,0);
\draw [color=black] (6,0)-- (6.32,-1);
\draw [color=black] (6.32,-1)-- (4.9,-2.32);
\draw [color=black] (4.9,-2.32)-- (4.28,-3.2);
\draw [color=black] (4.28,-3.2)-- (4.16,-5.34);
\draw [color=black] (4.16,-5.34)-- (3.82,-1.48);
\draw [color=black] (3.82,-1.48)-- (3.26,-0.02);
\draw [color=black] (3.26,-0.02)-- (2.4,1.04);
\draw [color=black] (2.4,1.04)-- (0.6,2.04);

\draw [dotted ] (3.78,5.12) -- (4.16,-5.34);
\draw [dashed ] (6.2,0.9)-- (4.28,-3.2);
\draw [dotted ] (0.6,2.04) -- (7.86,1.66) ;

\begin{scriptsize}
\fill [color=black] (0.6,2.04) circle (1.5pt);
\draw[color=black] (0.7,2.4) node {$v$};
\fill [color=black] (2.34,2.8) circle (1.5pt);
%\draw[color=black] (2.51,3.12) node {$B$};
\fill [color=black] (3.26,4.02) circle (1.5pt);
%\draw[color=black] (3.45,4.33) node {$C$};
\fill [color=black] (3.78,5.12) circle (1.5pt);
\draw[color=black] (3.96,5.54) node {$y$};
\fill [color=black] (4.3,3.82) circle (1.5pt);
%\draw[color=black] (4.47,4.13) node {$E$};
\fill [color=black] (5.74,2.42) circle (1.5pt);
%\draw[color=black] (5.9,2.73) node {$F$};
\fill [color=black] (7.86,1.66) circle (1.5pt);
\draw[color=black] (8.05,2.1) node {$x$};
\fill [color=black] (6.88,1.32) circle (1.5pt);
%\draw[color=black] (7.08,1.64) node {$H$};
\fill [color=black] (6.2,0.9) circle (1.5pt);
%\draw[color=black] (6.31,1.21) node {$I$};
\fill [color=black] (6,0) circle (1.5pt);
%\draw[color=black] (6.14,0.31) node {$J$};
\fill [color=black] (6.32,-1) circle (1.5pt);
\draw[color=black] (6.5,-0.68) node {$z$};
\fill [color=black] (4.9,-2.32) circle (1.5pt);
%\draw[color=black] (5.08,-2.01) node {$L$};
\fill [color=black] (4.28,-3.2) circle (1.5pt);
%\draw[color=black] (4.5,-2.88) node {$M$};
\fill [color=black] (4.16,-5.34) circle (1.5pt);
\draw[color=black] (4.5,-5.04) node {$u$};
\fill [color=black] (3.82,-1.48) circle (1.5pt);
%\draw[color=black] (4.01,-1.17) node {$O$};
\fill [color=black] (3.26,-0.02) circle (1.5pt);
%\draw[color=black] (3.43,0.29) node {$P$};
\fill [color=black] (2.4,1.04) circle (1.5pt);
%\draw[color=black] (2.58,1.35) node {$Q$};
\end{scriptsize}
\end{tikzpicture}
\caption{An illustration for the proof of Lemma \ref{distLemma} of a polygon containing the pseudo quadrilateral on points $x,y,v,u$ whose diagonals are drawn as dotted lines.}
\label{geoIntPic}
\end{figure}
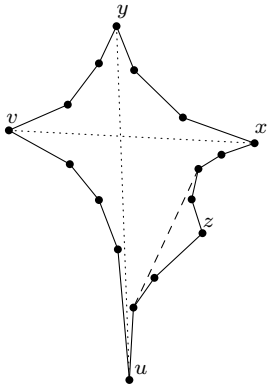

\begin{definition}
A \emph{pseudo triangle} is a polygon consisting of three convex vertices which are connected to each other by concave chains. A \emph{pseudo quadrilateral} is defined analogously on four convex vertices.\\
\end{definition}

\begin{lemma}
Let $P'$ be an arbitrary (not necessarily connected) subset of a polygon $P$, let $v$ be an arbitrary point in $P'$ and let $u \in P'$ be a point with maximum geodesic distance to $v$ (w.r.t. $P$) among all points in $P'$. It then holds for any  $x,y,z \in P'$ that $\max\{d(u,x),d(u,y), d(u,z)\} \geq \min\{d(x,y), d(x,z), d(y,z)\}$, with $d$ denoting the geodesic distance in $P$.
\label{distLemma}
\end{lemma}

\begin{proof}
For contradiction suppose that the claim in the lemma is false, i.e. let $x,y,z$ be three points with \\$\min\{d(x,y), d(x,z), d(y,z)\} > \max\{d(u,x),d(u,y), d(u,z)\}$. For two points $a, b \in P$, let $p(a,b)$ denote the shortest path in $P$ connecting them. It is easy to see that the path $p(u,v)$ connecting $u$ and $v$ intersects at most two of the three shortest paths among $x,y,z$, since they form a pseudo triangle. W.l.o.g. let $p(x,y)$ be a path not intersected by $p(u,v)$ and w.l.o.g. let $p(u,x)$ and $p(v,y)$ be the other two paths defining a pseudo quadrilateral together with $p(u,v)$ and $p(x,y)$. It then follows from Observation \ref{quadIneq} that $d(u,v) + d(x,y) \leq d(u,y) + d(v,x)$. On the other hand, since $d(v,x) \leq d(u,v)$ and $d(u,y) < d(x,y)$ by assumption, $d(u,v) + d(x,y) > d(u,y) + d(v,x)$, a contradiction. 
\end{proof}

%\begin{proof}
%Let $p(a,b)$ denote the shortest path between two points $a, b \in P$. The path $p(u,v)$ connecting the $u$ and $v$ 	can not only intersect one of the three shortest paths among $x,y,z$, since if it intersected only, say $p(x,y)$ and a point $a$ on $p(x,y)$ is the first point where $p(u,v)$ intersects $p(x,y)$ coming from $u$, then $d(a,v) < \min\{d(a,x),d(a,y)\}$ and thus $\max \{d(u,x), d(u,y)\} > d(u,v)$.  Now for contradiction suppose the claim in the lemma is false and let $x,y,z$ be three points with $\min\{d(x,y), d(x,z), d(y,z)\} > \max\{d(u,x),d(u,y), d(u,z)\}$. If $p(u,v)$ does not intersect any of the paths $p(x,y), p(x,z), p(y,z)$, then according to Observation \ref{quadIneq}, $d(u,v) + d(x,y) \leq d(u,y) + d(v,x)$.  On the other hand, since $d(v,x) \leq d(u,v)$ and $d(u,y) < d(x,y)$ by assumption, $d(u,v) + d(x,y) > d(u,y) + d(v,x)$, a contradiction. If $p(u,v)$ intersects exactly two paths $p(x,y)$ and $p(x,z)$ either in the interior (Figure \ref{geoIntPic}, left) or at a vertex (Figure \ref{geoIntPic}, right), then according to Observation \ref{quadIneq}, $d(u,v) + d(y,z) \leq d(u,z) +d(v,y)$. On the other hand, since $d(v,y) \leq d(u,v) = diam(P)$ and $d(u,z) < d(y,z)$ by assumption, $d(u,v) + d(y,z) > d(u,z) +d(v,y)$, a contradiction. 

%Observe that not all of $p(x,y), p(x,z)$ and $p(y,z)$ can consist of straight line segments since this would imply that $v$ lies inside $T$, a contradiction to $P$ being simple. W.l.o.g. assume that $p(x,y)$
%\end{proof}

\begin{obs}
In a pseudo quadrilateral, the sum of the lengths of the two diagonals is at least as large as the sum of the lengths of the two opposite sides.
\label{quadIneq}
\end{obs}

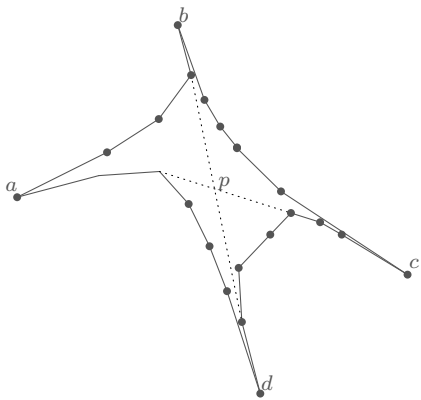
\begin{figure}[htbp]
\center
\definecolor{black}{rgb}{0.27,0.27,0.27}
\definecolor{black}{rgb}{0.33,0.33,0.33}
\begin{tikzpicture}[line cap=round,line join=round,>=triangle 45,x=.250cm,y=.25cm]
\clip(2.84,-10.26) rectangle (25.19,10.89);
\draw [color=black] (3.5,0.67)-- (8.28,3.06);
\draw [color=black] (8.28,3.06)-- (11.03,4.83);
\draw [color=black] (11.03,4.83)-- (12.75,7.17);
\draw [color=black] (12.75,7.17)-- (12.04,9.83);
\draw [color=black] (12.04,9.83)-- (13.46,5.85);
\draw [color=black] (13.46,5.85)-- (14.3,4.43);
\draw [color=black] (14.3,4.43)-- (15.19,3.32);
\draw [color=black] (15.19,3.32)-- (15.19,3.28);
\draw [color=black] (15.19,3.28)-- (17.53,0.98);
\draw [color=black] (17.53,0.98)-- (24.26,-3.45);
\draw [color=black] (24.26,-3.45)-- (20.76,-1.32);
\draw [color=black] (20.76,-1.32)-- (19.61,-0.66);
\draw [color=black] (19.61,-0.66)-- (18.06,-0.17);
\draw [color=black] (18.06,-0.17)-- (16.96,-1.32);
\draw [color=black] (16.96,-1.32)-- (15.28,-3.09);
\draw [color=black] (15.28,-3.09)-- (15.45,-5.97);
\draw [color=black] (15.45,-5.97)-- (16.43,-9.78);
\draw [color=black] (16.43,-9.78)-- (14.66,-4.33);
\draw [color=black] (14.66,-4.33)-- (13.73,-1.94);
\draw [color=black] (13.73,-1.94)-- (12.62,0.31);
\draw [color=black] (12.62,0.31)-- (11.07,2.04);
\draw [color=black] (11.07,2.04)-- (7.84,1.82);
\draw [color=black] (7.84,1.82)-- (3.5,0.67);
\draw [dotted] (12.75,7.17)-- (15.45,-5.97);
\draw [dotted] (11.07,2.04)-- (18.06,-0.17);
\begin{scriptsize}
\fill [color=black] (3.5,0.67) circle (1.5pt);
\draw[color=black] (3.2,1.24) node {$a$};
\fill [color=black] (8.28,3.06) circle (1.5pt);
%\draw[color=black] (8.59,3.63) node {$B$};
\fill [color=black] (11.03,4.83) circle (1.5pt);
%\draw[color=black] (11.38,5.4) node {$C$};
\fill [color=black] (12.75,7.17) circle (1.5pt);
%\draw[color=black] (12.21,7.35) node {$b'$};
\fill [color=black] (12.04,9.83) circle (1.5pt);
\draw[color=black] (12.35,10.41) node {$b$};
\fill [color=black] (13.46,5.85) circle (1.5pt);
%\draw[color=black] (13.77,6.42) node {$F$};
\fill [color=black] (14.3,4.43) circle (1.5pt);
%\draw[color=black] (14.66,5.01) node {$G$};
\fill [color=black] (15.19,3.32) circle (1.5pt);
%\draw[color=black] (15.54,3.9) node {$H$};
\fill [color=black] (15.19,3.28) circle (1.5pt);
%\draw[color=black] (15.41,3.86) node {$I$};
\fill [color=black] (17.53,0.98) circle (1.5pt);
%\draw[color=black] (17.8,1.55) node {$J$};
\fill [color=black] (24.26,-3.45) circle (1.5pt);
\draw[color=black] (24.61,-2.87) node {$c$};
\fill [color=black] (20.76,-1.32) circle (1.5pt);
%\draw[color=black] (21.07,-0.75) node {$L$};
\fill [color=black] (19.61,-0.66) circle (1.5pt);
%\draw[color=black] (19.39,-1.15) node {$M$};
\fill [color=black] (18.06,-0.17) circle (1.5pt);
%\draw[color=black] (18.42,-0.6) node {$c'$};
\fill [color=black] (16.96,-1.32) circle (1.5pt);
%\draw[color=black] (17.31,-0.75) node {$O$};
\fill [color=black] (15.28,-3.09) circle (1.5pt);
%\draw[color=black] (15.59,-2.52) node {$P$};
\fill [color=black] (15.45,-5.97) circle (1.5pt);
%\draw[color=black] (16.1,-5.4) node {$d'$};
\fill [color=black] (16.43,-9.78) circle (1.5pt);
\draw[color=black] (16.78,-9.2) node {$d$};
\fill [color=black] (14.66,-4.33) circle (1.5pt);
%\draw[color=black] (14.97,-3.76) node {$S$};
\fill [color=black] (13.73,-1.94) circle (1.5pt);
%\draw[color=black] (14.08,-1.37) node {$T$};
\fill [color=black] (12.62,0.31) circle (1.5pt);
\draw[color=black] (14.5,1.4) node {$p$};
\end{scriptsize}
\end{tikzpicture}

\caption{Illustration of the proof of Observation \ref{quadIneq}.}
\label{figproof}
\end{figure}
 
\begin{proof}
Let $a,b,c,d$ be  the four convex vertices of a pseudo quadrilateral and let $p$ be the intersection point of the two diagonal shortest paths $p(a,c)$ and $p(b,d)$ as shown in Figure \ref{figproof}. It is well known that in any simple polygon the shortest paths between any three points form a pseudo triangle. Thus, according to Observation \ref{triIneq}, $d(a,p) + d(p,b) \geq d(a,b)$ and  $d(c,p) + d(p,d) \geq d(c,d)$. Thus  $d(a,p) + d(p,b) + d(c,p) + d(p,d) \geq  d(a,b) + d(c,d)$. The inequality for $d(a,d)$ and $d(b,c)$ can be shown analogously and thus the observation follows. 
\end{proof}

The following result is known (see for example \cite{geoDiaSimp}) and we thus omit a proof here.
\begin{obs}[Triangle Inequality]
In any pseudo triangle the sum of the lengths of two paths connecting two convex vertices is at least as large as the length of the remaining shortest path.
\label{triIneq}
\end{obs}

\begin{theorem}
Geodesic Unit Disk Packing is $\mathsf{NP}$-hard in polygons with holes.
\label{unitPackHard}
\end{theorem}

\begin{proof} 
In \cite{garey1977rectilinear} it is shown that the Maximum Independent Set problem on planar graphs of maximum degree $3$ is $\mathsf{NP}$-hard. Given such an instance $G=(V,E)$, it is easy to see that replacing an edge by a path of odd length $l$ increases the maximum independent set by $(l-1)/2$. We now reduce an instance of this problem to our problem, by orthogonally embedding $G$ in the plane on an integer grid \cite{gridEmd}. We then replace each edge $e \in E$ by a path of $l_e$ straight line edges, each edge is of length $1$, with $l_e$ being an odd number. We denote the obtained graph by $G'=(V',E')$. Denoting the resulting embedding by $P = P(G')$, we attach a polygonal chain of length $1.5$ at the midpoint of each edge. Using the arguments of the proof of Theorem 1 in \cite{penVig}, it follows that there is enough space in the embedding $P$ to attach the polygonal paths on each edge separately. Now any optimal packing centers a geodesic unit disk at the end of each polygonal chain and packs the remaining polygon (whose free space can now only be packed by placing disks at the vertices $V'$) with a maximum independent set for $G'$. Therefore, $G$ has an independent set of size at most $M$ if and only if $P$ has a packing of size at most $\sum_{e \in E} (l_e-1)/2+M+|E'|$.
\end{proof}

\section{Geodesic Unit Disk Covering}
\label{udCover}

In this section we show that covering a polygon with holes using the minimum number of geodesic unit disks is $\mathsf{NP}$-hard.

\begin{problem}[Geodesic Unit Disk Covering]
Given a polygon $P$, find a cover of $P$ with fewest geodesic unit disks.
\label{geoUnCov}
\end{problem}

\begin{theorem}
Geodesic Unit Disk Covering is $\mathsf{NP}$-hard in polygons with holes.
\label{geoUnCovHard}
\end{theorem}

\begin{proof}
We reduce an instance $G = (V,E)$ of the $\mathsf{NP}$-hard vertex cover problem on planar graphs with maximum degree $3$ (see \cite{Garey:1979:CIG:578533}) to it. It is easy to see that replacing an edge by a path of odd length $l$  increases the size of a vertex cover by $(l-1)/2$. We now reduce such a vertex cover instance to the unit disk cover problem by orthogonally embedding $G$ in the plane on an integer grid \cite{gridEmd}. We then replace each edge $e \in E$ by a path of  $l_e$ straight line edges each length $1$, with $l_e$ being an odd number, thus obtaining a new graph $G'=(V',E')$. We then replace each edge by an edge gadget as shown in Figure \ref{covGad} and call the resulting polygon with holes $P=P(G')$. For an edge $\{u,v\} \in E'$ the corresponding gadget consists of two paths of length $1$ connecting $u$ and $v$. In the middle of each path, an additional path of length $0.5$ is attached. It thus follows that an edge gadget is completely covered by a single unit disk if and only if it the disk is centered at either $u$ or $v$. Thus $G$ has a vertex cover of size at most $M$ if and only if $P$ has a covering of size at most $\sum_{e \in E} (l_e-1)/2+ M$. Furthermore, using the arguments of the proof of Theorem 1 in \cite{penVig}, it is clear that there is enough space in the embedding $P$ to  replace edges in $E'$ by their gadgets without overlapping any other gadgets.
\end{proof}

\begin{figure}[htp]
\center
\definecolor{qqqqff}{rgb}{0.33,0.33,0.33}
\begin{tikzpicture}[line cap=round,line join=round,>=triangle 45,x=1.4cm,y=1.4cm]
\clip(1.91,1.41) rectangle (6.3,2.7);
\draw (2,2)-- (4,1.5);
\draw (4,1.5)-- (6.06,1.98);
\draw (6.06,1.98)-- (4,2.5);
\draw (4,2.5)-- (2,2);
\draw (4,2.5)-- (4,2.2);
\draw (4,2.2)-- (3,2.2);
\draw (3,2.2)-- (3,2.1);
\draw (3,2.1)-- (4,2.1);
\draw (4,1.5)-- (4,1.8);
\draw (4,1.8)-- (3,1.8);
\draw (3,1.8)-- (3,1.9);
\draw (3,1.9)-- (4,1.9);
\begin{scriptsize}
\fill [color=qqqqff] (2,2) circle (1.5pt);
\draw[color=qqqqff] (2.08,2.14) node {$u$};
\fill [color=qqqqff] (6.06,1.98) circle (1.5pt);
\draw[color=qqqqff] (6.14,2.12) node {$v$};
\end{scriptsize}
\end{tikzpicture}
\caption{An edge gadget for an edge $\{u,v\}$ as used in the proof of Theorem \ref{geoUnCovHard}.}
\label{covGad}
\end{figure}
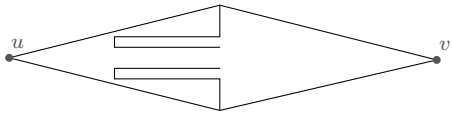

\section{Geodesic $k$-Covering}
\label{geodesicKcov}

In this section we present a $2$-approximation algorithm for covering a polygon, possibly with holes, using $k$ geodesic disks and show that the problem is $\mathsf{NP}$-hard.

\begin{problem}[Geodesic $k$-Covering]
Given a polygon $P$, possibly with holes, find a cover of $P$ with $k$ geodesic disks whose maximal radius is minimized. 
\label{geoCov}
\end{problem}

\begin{problem}[Metric $k$-Clustering]
Given a set $S$ of $n$ points from a metric space, find $k$ smallest disks, such that they cover all points in $S$.
\end{problem}

Since the geodesic distance is a metric (see also Observation \ref{triIneq}), it follows that the \emph{Geodesic $k$-Covering} problem can be stated as a \emph{Metric $k$-clustering} problem, which, although $\mathsf{NP}$-hard, can be approximated within a factor of two \cite{kcenterRef}. On the other hand, the fact that a polygon contains an uncountable number of points requires some further reasoning as to why the Geodesic $k$-Covering problem is of discrete nature.   \\

The following algorithm simply places $k$ points inside a polygon, possibly with holes, such that the minimum geodesic distance among them is maximized.

\begin{algorithm}
\begin{algorithmic}
\caption{gonzalezPlacement($P,k$)}
\STATE $C \leftarrow \{c_1\}$ \COMMENT{with $c_1$ an arbitrary point in $P$}
\STATE $\mathcal{M} \leftarrow \emptyset$ \COMMENT{sequence of shortest path maps}
\FOR{$i \leftarrow 2$ to $k$} 
\STATE compute the shortest path map $M_{i-1}$ for $c_{i-1}$  
\STATE $\mathcal{M}  \leftarrow \mathcal{M}  \cup \{M_{i-1}\}$ \COMMENT{store it in collection $\mathcal{M}$}
\STATE compute  the geodesic voronoi diagram of $C$ in $P$
\STATE compute $A_i = A_i^1 \cup A_i^2 \cup vertices(P)$ 
\COMMENT{as described in the proof of Theorem \ref{2approxGonThm}}
\STATE $c_i \leftarrow \arg\max_{a \in A}\min_{c \in C} d(c,a)$ using $M_j$, with $1\leq j \leq i$
\STATE $C \leftarrow C \cup \{c_i\}$ % \COMMENT{with $c_i$ an arbitrary point in $A$}
\ENDFOR
\RETURN $C$
\end{algorithmic}
\end{algorithm}

%[TODO: computing largest empty circles with location constraint might be used as an algorithm to get rid of a k factor.]

\begin{theorem}
The $gonzalezPlacement$ algorithm finds a geodesic cover of a polygon $P$ (possibly with holes) whose maximum radius is at most twice as large as the largest radius in an optimal cover in time $O(k^2(n+k)\log(n+k))$.% [maybe $O(kn+(n+k)\log(n+k))$]  2( maybe faster in simple plygons
\label{2approxGonThm}
\end{theorem}
\begin{proof}
The approximation ratio follows from the fact that Metric $k$-Clustering can be approximated within a factor of two (see \cite{kcenterRef}). 
In order to prove the running time we need to bound the time to find the next center $c_{i+1}$. For this let $C_i$ denote the set of centers after the $i$-th iteration. We find the next center in a polygon, possibly with holes, by computing the geodesic voronoi diagram \cite{AronovGeo} of $C_{i}$ in $P$ using the continuous Dijkstra paradigm (see also \cite{higherGeoVorHoles}) in time  $O((n+i)\log(n+i))$ \cite{Hershberger97anoptimal}. The center $c_{i+1}$ is either contained in the set of vertices of the voronoi diagram (denoted by $A^1_i$) or lies at the intersection points of the voronoi edges with the boundary of $P$ (denoted by $A^2_i$) or in the set of vertices of $P$ (denoted by $vertices(P)$). This holds since any other point is incompletely constrained and its distance to $C_i$ could thus be enlarged \cite{Shamos75}. We denote the union $A^1_i \cup A^2_i \cup vertices(P)$ of all points of interest by $A_i$ and remark that the complexity of the geodesic voronoi diagram is $O(n+i)$ (see \cite{AronovGeo}), implying $|A_i| =  O(n+i)$. Next, we compute the shortest path map $M_i$ for $c_i$ in time $O(n \log n)$ using again the algorithm of \cite{Hershberger97anoptimal} and add it to the collection $\mathcal{M}$ of shortest path maps. Finding the point  furthest from any point in $C_i$, i.e. $\arg\max_{a \in A}\min_{c \in C_i} d(c,a)$, can be done by using each of the $i$  shortest path maps in $\mathcal{M}$ in total time $O(i(i+n)\log n)$ per iteration. \end{proof}

\begin{theorem}
Geodesic $k$-Covering is $\mathsf{NP}$-hard.
\label{hardCov}
\end{theorem}
\begin{proof} 
From the proof of Theorem \ref{geoUnCovHard} it clearly follows that the decision version of the Geodesic Unit Disk Covering problem, i.e. whether a given polygon with holes can be covered with $k$ geodesic unit disks, is $\mathsf{NP}$-hard. Since solving the $k$-Covering problem and checking whether the minimum radius is at most one decides the Geodesic Unit Disk Covering problem, the theorem follows.
\end{proof}

%\begin{proof}
%Noting that any planar graph can be modeled as polygon with holes the claim follows by Theorem 3.1 in \cite{KarivHakimi1979}. 
%This theorem shows that the absolute $p$-center of a graph, i.e. a set of points (not necessarily at vertices) which minimizes the maximum distance of any vertex to this set is hard to compute.
%\end{proof}

\section{Geodesic $k$-Packing}
\label{geoKpacking}

In this section we present a $4$-approximation algorithm for packing $k$ geodesic disks into a polygon, possibly with holes, and show that the problem is $\mathsf{NP}$-hard.

\begin{problem}[Geodesic $k$-Packing]
Given a polygon $P$, possibly with holes, pack $k$ geodesic disks whose minimum radius is maximized.
\end{problem}

\begin{theorem}
The $gonzalezPlacement$ algorithm finds a geodesic packing whose minimum radius is at least a fourth of the minimum  radius in an optimal solution.
\end{theorem}
\begin{proof}
Letting $r$ denote the largest radius needed to cover $P$ with geodesic disks centered at the points returned by $gonzalezPlacement$, it is well known (see \cite{kcenterRef}) that all centers found by $gonzalezPlacement$ have distance of at least $r$. Thus centering disks of radius $r/2$ at these points provides a $k$-Packing, which, according to Proposition \ref{pack2cov} has a minimum radius of at least a fourth of the optimal solution.
\end{proof}

\begin{prop}
Let $s^*$ denote the optimal radius of the Geodesic $k$-Packing problem and let $r$ be the maximum radius of an arbitrary Geodesic $k$-Cover. It then holds that $s^* \leq 2r$.
\label{pack2cov}
\end{prop}

\begin{proof}
%, since the radius $\tilde{r}$ obtained from  is at least $r^*$, and thus $s^* \leq 2\tilde{r}$. 

For contradiction suppose that $s^* > 2r$ and let $C$ be the centers of a packing achieving such a minimum radius. Let $q \in P$ be a point lying halfway between two centers $c,c'\in C$ for which $d(c,c') = 2s^*$.  Observe that by the (geodesic) triangle inequality, $\min_{c \in C} d(q,c) \geq s^* > 2r$ and thus no geodesic radius $r$ disk containing $q$ can contain any element of $C$. But this contradicts the fact that the polygon can be covered by $k$ disks of maximum radius $r$, since each of the $k$ points in $C$ need to be contained in a disk and no disk can contain more than one point of $C$. 
\end{proof}

Using a similar argument as for the proof of Theorem \ref{hardCov}, we obtain the following Theorem.

\begin{theorem}
Geodesic $k$-Packing is $\mathsf{NP}$-hard.
\end{theorem}
%\begin{proof} 
%From Theorem \ref{unitPackHard} it clearly follows that the decision version of the Geodesic Unit Disk Packing problem, i.e. whether a given polygon can be packed with $k$ geodesic unit disks, is $\mathsf{NP}$-hard. Since solving the $k$-Packing problem and checking whether the minimum radius is at least one decides the unit disks packing problem, the theorem follows.
%\end{proof}

\section{Exact Covering with Two Geodesic Disks}
\label{exact2cov}

In this section we are studying the problem of covering a simple polygon $P$ with two geodesic disks of minimum maximal radius. We solve this problem by first considering the decision version, i.e. whether $P$ can be covered with two radius $r$ disks. We then apply parametric search \cite{Megiddo83applyingparallel} on the decision algorithm in order to solve the minimization problem.

%compute a set of candidate pairs of vertices of $P$. Let $p,q$ be such a pair and let $C_p$ and $C_q$ be the two geodesic circles of radius $r$ centered at $p, q$ respectively. According to Lemma \ref{noReflex} it suffices to consider pairs of convex vertices. 

Following \cite{Shin98computingtwo}, the basic idea for solving the decision problem is to first compute an arrangement $\mathcal{C}$ of geodesic radius $r$ circles, each centered at a convex vertex of $P$ in time $O(n^2)$ using the algorithm in \cite{Guibas} (see also \cite{geodPac}). If $P$ can be covered by two geodesic disks of radius $r$ then it can be covered with two such disks centered on arcs of the circles in $\mathcal{C}$. This can be seen to hold, by noting that such a configuration minimizes the distance between the centers of the two covering disks. The $testTwoDiskCover$ algorithm now tests for each arc pair in $\mathcal{C}$ whether two disks centered at these arcs cover all vertices of $P$. It is easy to see that independently of where in an arc of $\mathcal{C}$ a radius $r$ disk is centered, it always contains the same vertices of $P$. Furthermore, Lemma \ref{twoEdges} ensures that if all vertices are covered, then there are at most two uncovered edges. Using the algorithm of \cite{Shin98computingtwo}, one can check in constant time if the at most two uncovered edges can be covered with two disks centered in the same arcs as the current disks. Lastly, according to Lemma \ref{boundInt}, a completely covered boundary implies a covering of the interior and thus correctness of the $testTwoDiskCover$ algorithm follows. Furthermore, it is easy to see that the running time of $testTwoDiskCover$ is $O(n^5)$. Thus applying parametric search on the decision problem results in an algorithm for finding two geodesic disks of minimum maximal radius covering $P$ which runs in time $O(n^8 \log n)$.

\begin{algorithm}
\begin{algorithmic}
\STATE $\forall v \in V_{convex}(P)$ : center a geodesic $r$-circle $C_v$ on $v$ \COMMENT{$V_{convex}(P)$ denoting the set of convex vertices of $P$}
\STATE build circle arrangement $\mathcal{C}$
\FOR{$\{p,q\} \in V_{convex} \times V_{convex}$ with $p \neq q$} 
\FOR{Arc $A \in C_p$ of $\mathcal{C}$}  \STATE \COMMENT{$C_p$ denoting the geodesic $r$-circle centered at $p$.}
\STATE center a disk $D_a$ at some $a \in A$ 
\FOR{Arc $B \in C_q$ of $\mathcal{C}$} 
\STATE center a disk $D_b$ at some $b \in B$ 
\IF{$(V(P) \not\subseteq D_a \cup D_b)$}
\STATE  \COMMENT{$D_a$ and $D_b$ do not cover all vertices.}
\STATE continue
\ENDIF
\STATE Let $E_{unvoc}$ be the $\leq 2$ uncovered edges.
\IF{($\exists a' \in A, b' \in B$ s.t. disks $D_{a'}$ and $D_{b'}$ cover $E_{unvoc}$)}
\RETURN $true$
\ENDIF
\ENDFOR
\ENDFOR
\ENDFOR
\RETURN $false$
\end{algorithmic}
\caption{testTwoDiskCover($P,r$)}
\label{algoTest}
\end{algorithm}

\begin{lemma}
If two geodesic disks cover all vertices of a simple polygon then there are at most two uncovered edges.
\label{twoEdges}
\end{lemma}
\begin{proof}
By the triangle inequality it follows that if a geodesic disk covers both endpoints of an edge, it also covers its interior. Therefore, any uncovered edge has its endpoints covered by two different disks and thus there has to be a point in the interior of the edge which is equidistant from the two disk centers. Since in \cite{AronovGeo} it is shown that the geodesic bisector between two points has exactly two points on the polygon boundary, there can be at most two such edges.
\end{proof}

\begin{lemma}
If two geodesic disks cover the boundary of a simple polygon, then they also cover its interior.
\label{boundInt}
\end{lemma} 
\begin{proof}
Let $D_1$ and $D_2$ be two geodesic disks centered at $c_1$, $c_2$ respectively which, w.l.o.g., both lie on the $x$-axis and which cover the boundary $\partial P$. Assume for contradiction that there is a point $p \in P \setminus (D_1 \cup D_2)$ and let $s$ be the vertical line segment through $p$ which ends in $\partial P$ at the point $r$ and $s$. The shortest paths among $c_1, c_2, r$ and $s$ span a pseudo quadrilateral containing $p$ in its interior and thus by the triangle inequality $\max\{d(c_1, r), d(c_1, s)\} \geq d(c_1, p)$ and $\max\{d(c_2, r), d(c_2, s)\} \geq d(c_2, p)$ contradicting that the boundary is covered.
\end{proof}

\section{Acknowledgments}

The author would like to thank Peter Bra{\ss} and Jon Lenchner for introducing the problems to him and Ning Xu and Peter Bra{\ss} for helpful discussions.

\bibliographystyle{abbrv}	% (uses file "plain.bst")
\bibliography{ref}

\end{document}